\documentclass[runningheads,11pt]{llncs}

\usepackage{geometry}
\usepackage{amsfonts}
\usepackage[utf8]{inputenc}
\usepackage[T1]{fontenc}
\usepackage{lmodern}
\usepackage{amsmath}
\usepackage{caption} 
\usepackage{hyperref}
\hypersetup{
	colorlinks,
	citecolor=black,
	filecolor=black,
	linkcolor=black,
	urlcolor=black
}
\usepackage{algorithm2e}
\usepackage{array}
\newtheorem{cor}{Corollary}
\usepackage{bm}
\usepackage{spverbatim}
\usepackage{breakcites}
\usepackage{enumitem}
\usepackage{hhline}
\usepackage{mathrsfs}
\usepackage{enumitem}
\usepackage{multicol,xparse}
\usepackage{changepage}
\usepackage{mdframed} 
\usepackage{nccmath}
\usepackage[sanserif,basic]{complexity}
\newtheorem{observation}{Observation}

\DeclareMathOperator*{\myplus}{\overline{\bigoplus}}

\newcolumntype{P}[1]{>{\centering\arraybackslash}p{#1}}

\author{Vipin Singh Sehrawat \inst{1}, Yvo Desmedt \inst{1,2}}
\institute{Department of Computer Science, The University of Texas at Dallas, USA\\
	\and 
	Department of Computer Science,
	University College London, UK
}
\title{Bi-Homomorphic Lattice-Based PRFs \\ and Unidirectional Updatable Encryption\thanks{This is the full version of the paper that appears in Y. Mu et al. (Eds.): CANS 2019, LNCS 11829, pp. 1–21. DOI: {10.1007/978-3-030-31578-8\_1}~\cite{Vipin[19]}.}}
\titlerunning{Bi-Homomorphic Lattice-Based PRFs and Unidirectional Updatable Encryption}
\authorrunning{V. S. Sehrawat and Y. Desmedt}

\begin{document}
\maketitle

\begin{abstract}
\normalsize We define a pseudorandom function (PRF) $F: \mathcal{K} \times \mathcal{X} \rightarrow \mathcal{Y}$ to be bi-homomorphic when it is fully Key homomorphic and partially Input Homomorphic (KIH), i.e., given $F(k_1, x_1)$ and $F(k_2, x_2)$, there is an efficient algorithm to compute $F(k_1 \oplus k_2, x_1 \ominus x_2)$, where $\oplus$ and $\ominus$ are (binary) group operations. The homomorphism on the input is restricted to a fixed subset of the input bits, i.e., $\ominus$ operates on some pre-decided $m$-out-of-$n$ bits, where $|x_1| = |x_2| = n$, and the remaining $n-m$ bits are identical in both inputs. In addition, the output length, $\ell$, of the operator $\ominus$ is not fixed and is defined as $n \leq \ell \leq 2n$, hence leading to Homomorphically induced Variable input Length (HVL) as $n \leq |x_1 \ominus x_2| \leq 2n$. We present a learning with errors (LWE) based construction for a HVL-KIH-PRF family. Our construction is inspired by the key homomorphic PRF construction due to Banerjee and Peikert (Crypto 2014). 

An updatable encryption scheme allows rotations of the encryption key, i.e., moving existing ciphertexts from old to new key. These updates are carried out via \emph{update tokens}, which can be used by an untrusted party since the update procedure does not involve decryption of the ciphertext. We use our novel PRF family to construct an updatable encryption scheme, named QPC-UE-UU, which is quantum-safe, post-compromise secure and supports unidirectional ciphertext updates, i.e., the update tokens can be used to perform ciphertext updates but they cannot be used to undo already completed updates. Our PRF family also leads to the first left/right key homomorphic constrained-PRF family with HVL.
\end{abstract}

\section{Introduction}
In a PRF family~\cite{Gold[86]}, each function is specified by a short, random key, and can be easily computed given the key. Yet the function behaves like a random one, in the sense that if you are not given the key, and are computationally bounded, then the input-output behavior of the function looks like that of a random function. Since their introduction, PRFs have been one of the most fundamental building blocks in cryptography. For a PRF $F_s$, the index $s$ is called its key or seed. Many variants of PRFs with additional properties have been introduced and have found a plethora of applications in cryptography. \\[1.5mm]
\textbf{Key Homomorphic (KH) PRFs:} A PRF family $F$ is KH-PRF if the set of keys has a group structure and if there is an efficient algorithm that, given $F_s(x)$ and $F_t(x)$, outputs $F_{s+t}(x)$~\cite{Naor[99]}. Multiple KH-PRF constructions have been proposed via varying approaches~\cite{Naor[99],Boneh[13],Ban[14],Parra[16]}. These functions have many applications in cryptography, such as, symmetric-key proxy re-encryption, updatable encryption, and securing PRFs against related-key attacks~\cite{Boneh[13]}. But, lack of input homomorphism limits the privacy preserving applications of KH-PRFs. For instance, while designing solutions for searchable symmetric encryption~\cite{Curt[06]}, it is highly desirable to hide the search patterns, which can be achieved by issuing random queries for each search. But, this feature cannot be supported if the search index is built by using a KH-PRF family, since it would require identical query (i.e., function input) in order to perform the same search. \\[1.5mm]
\textbf{Constrained PRFs (CPRFs):} Constrained PRFs (also called delegatable PRFs) are another extension of PRFs. They enable a proxy to evaluate a PRF on a strict subset of its domain using a trapdoor derived from the CPRF secret key. A trapdoor is constructed with respect to a certain policy predicate that determines the subset of the input values for which the proxy is allowed to evaluate the PRF. Introduced independently by Kiayias et al.~\cite{Kia[13]}, Boneh et al.~\cite{Rent[13]} and Boyle et al.~\cite{Boyle[14]} (termed functional PRFs), CPRFs have multiple interesting applications, including broadcast encryption, identify-based key exchange, batch query supporting searchable symmetric encryption and RFID. Banerjee et al.~\cite{Ban[15]}, and Brakerski and Vaikuntanathan~\cite{Bra[15]} independently introduced KH-CPRFs.\\[1.5mm] 
\textbf{Variable Input Length (VIL) PRFs:} VIL-PRFs~\cite{Bell[96]} serve an important role in constructing variable length block ciphers~\cite{Mihir[99]} and authentication codes~\cite{Bell[96]}, and are employed in prevalent protocols like Internet Key Exchange (IKEv2). No known CPRF or KH-CPRF construction supports variable input length. \\[2mm] 
\textbf{Updatable Encryption (UE):} In data storage, key rotation refers to the process of (periodically) exchanging the cryptographic key material that is used to protect the data. Key rotation is a desirable feature for cloud storage providers as it can be used to revoke old keys, that might have been comprised, or to enforce data access revocation. All major cloud storage providers (eg. Amazon's Key Management Service~\cite{AWS[00]}, Google Cloud Platform~\cite{Goo[00]}) implement some variants of data-at-rest encryption and hybrid encryption techniques to perform key rotation~\cite{Ever[17]}, which although efficient, do not support full key rotation as the procedures are designed to only update the key encapsulation but not the long-term key. Symmetric updatable encryption, introduced by Boneh et al.~\cite{Boneh[13]} (BLMR henceforth), supports full key rotation without performing decryption, i.e., the ciphertexts created under one key can be securely updated to ciphertexts generated via another key, with the help of a re-encryption/update token. 

Everspaugh et al.~\cite{Ever[17]} pointed out that the UE scheme from BLMR addresses relatively weak confidentiality goals and does not even consider integrity. They proposed a new security notion, named re-encryption indistinguishability, to better capture the idea of fully refreshing the keys upon rotation. They also presented an authenticated encryption based UE scheme, that satisfies the requirements of their new security model. Recently, Lehmann and Tackmann~\cite{Anja[18]} observed that the previous security models/definitions (from BLMR and Everspaugh et al.) do not capture post-compromise security, i.e., the security guarantees after a key compromise. They also proved that neither of the two schemes is post-compromise secure under adaptive attacks. In the same paper, they presented the first UE scheme with post-compromise security. However, the security of that scheme is based on the Decisional Diffie Hellman (DDH) assumption, rendering it vulnerable to quantum computers~\cite{Shor[94]}. It is important to note that all existing UE schemes only support bidirectional updates, i.e., the update token used to refresh the ciphertext for the current epoch can also be used to revert back to the previous epoch's ciphertext. Naturally, this is an undesirable feature. Hence, unidirectional updates~\cite{Anja[18]}, where the update tokens cannot be used to undo the ciphertext updates, is a highly desirable feature for UE schemes. But as mentioned earlier, no existing UE scheme achieves unidirectional updates.  

\subsection{Our Contributions}
Our contributions can be classified into the following two broad classes. 
\begin{enumerate}
	\item \textbf{Novel PRF Classes and Constructions.} We introduce fully Key homomorphic and partially Input Homomorphic (KIH) PRFs with Homomorphically induced Variable input Length (HVL). A PRF, $F$, from such function family satisfies the condition that given $F_{k_1}(x_1)$ and $F_{k_2}(x_2)$, there exists an efficient algorithm to compute $F_{k_1 \oplus k_2}(x_1 \ominus x_2)$, where $|x_1 \ominus x_2| \geq |x_1|~(= |x_2|)$ and the input homomorphism effects only some fixed $m$-out-of-$n$ bits. We present a Learning with Errors (LWE)~\cite{Reg[05]} based construction for such a PRF family. Our construction is inspired by the KH-PRF construction from Banerjee and Peikert~\cite{Ban[14]}. A restricted case of our PRF family leads to another novel PRF class, namely left/right KH-CPRF with HVL.	
	\item \textbf{Quantum-Safe Post-Compromise Secure UE Scheme with Unidirectional Updates.} We use our HVL-KIH-PRF family to construct the first quantum-safe, post-compromise secure updatable encryption scheme with unidirectional updates. We know that the KH-PRF based UE scheme from BLMR is not post-compromise secure because it never updates the nonce~\cite{Anja[18]}. Since our HVL-KIH-PRF family supports input homomorphism, in addition to key homomorphism, it allows updates to the nonce (i.e., the PRF input). Hence, we turn the BLMR UE scheme post-compromise secure by replacing their KH-PRF by our HVL-KIH-PRF, and introducing randomly sampled nonces for the cihpertext updates. The bi-homomorphism of our PRF family also allows us to enforce unidirectional updates. Since our PRF construction is based on the learning with errors (LWE) problem~\cite{Reg[05]}, our UE scheme is also quantum-safe.
\end{enumerate}

\subsection{Organization} Section~\ref{sec2} recalls the necessary background for the rest of the paper. Section~\ref{sec4} introduces and defines KIH-PRF, HVL-KIH-PRF and HVL-KH-CPRF. In Section~\ref{sec5}, we present a LWE-based construction for a HVL-KIH-PRF family, provide its proof of correctness and discuss the different types of input homomorphisms that it supports. Section~\ref{sec6} gives the security proof for our HVL-KIH-PRF construction, while Section~\ref{sec7} analyzes its time complexity. Section~\ref{sec8} presents the construction of left/right HVL-KH-CPRF, that follows from a restricted case of our HVL-KIH-PRF family. In Section~\ref{sec9}, we use of our HVL-KIH-PRF family to construct the first quantum-safe, post-compromise secure updatable encryption scheme with unidirectional updates. Section~\ref{sec10} discusses an interesting open problem, solving which would lead to a novel, search pattern hiding searchable encryption scheme. Section~\ref{sec11} gives the conclusion. 

\section{Background}\label{sec2}
This section recalls the necessary definitions required for the rest of the paper.

\begin{definition}[Negligible Function]
	\emph{For security parameter $\omega$, a function $\epsilon(\omega)$ is called \textit{negligible} if for all $c > 0$ there exists a $\omega_0$ such that $\epsilon(\omega) < 1/\omega^c$ for all $\omega > \omega_0$.}
\end{definition}

\subsection{Learning with Errors}\label{LWE}
The learning with errors (LWE) problem requires to recover a secret $s$ given a sequence of `approximate' random linear equations on it. LWE is known to be hard based on certain assumptions regarding the worst-case hardness of standard lattice problems such as GapSVP (decision version of the Shortest Vector Problem) and SIVP (Shortest Independent Vectors Problem)~\cite{Reg[05],Pei[09]}. Many cryptosystems have been constructed whose security can be proven under the LWE problem, including (identity-based, leakage-resilient, fully homomorphic, functional) encryption~\cite{Reg[05],Gen[08],Adi[09],Reg[10],Shweta[11],Vinod[11],Gold[13]}, oblivious transfer~\cite{Pei[08]}, (blind) signatures~\cite{Gen[08],Vad[09],Markus[10],Vad[12]}, PRFs~\cite{Ban[12]}, KH-PRFs~\cite{Boneh[13],Ban[14]}, KH-CPRFs~\cite{Ban[15],Bra[15]}, hash functions~\cite{Katz[09],Pei[06]}, etc. 

\begin{definition}[Decision-LWE~\cite{Reg[05]}]\label{decisionLWE} 
	\emph{For positive integers $n$ and $q \geq 2$, and an error (probability) distribution $\chi = \chi(n)$ over $\mathbb{Z}_q$, the decision-LWE${}_{n, q, \chi}$ problem is to distinguish between the following pairs of distributions: 
		\[(\textbf{A}, \textbf{A}^T \textbf{s} + \textbf{e}) \quad \text{and} \quad (\textbf{A}, \textbf{u}),\] 
		where $m = \poly(n), \textbf{A} \xleftarrow{\; \$ \;} \mathbb{Z}^{n \times m}_q, \textbf{s} \xleftarrow{\; \$ \;} \mathbb{Z}^n_q, \textbf{e} \xleftarrow{\; \$ \;} \chi^m,$ and $\textbf{u} \xleftarrow{\; \$ \;} \mathbb{Z}^m_q$.}
\end{definition} 

\begin{definition}[Search-LWE~\cite{Reg[05]}]\label{searchLWE}
	\emph{For positive integers $n$ and $q \geq 2$, and an error (probability) distribution $\chi = \chi(n)$ over $\mathbb{Z}_q$, the search-LWE${}_{n, q, \chi}$ problem is to recover $\textbf{s} \in \mathbb{Z}^n_q$, given $m (= \poly(n))$ independent samples of $(\textbf{A}, \textbf{A}^T \textbf{s} + \textbf{e})$, where $\textbf{A} \xleftarrow{\; \$ \;} \mathbb{Z}^{n \times m}_q, \textbf{s} \xleftarrow{\; \$ \;} \mathbb{Z}^n_q,$ and $\textbf{e} \xleftarrow{\; \$ \;} \chi^m.$}
\end{definition}

Regev~\cite{Reg[05]} showed that for a certain noise distribution $\chi$ and a sufficiently large $q$, the LWE problem is as hard as the worst-case SIVP (Shortest Independent Vectors Problem) and GapSVP (decision version of the Shortest Vector Problem) under a quantum reduction (see also~\cite{Pei[09],Bra[13]}). These results have been extended to show that $\textbf{s}$ can be sampled from a low norm distribution (in particular, from the noise distribution $\chi)$ and the resulting problem is as hard as the basic LWE problem~\cite{Benny[09]}. Similarly, the noise distribution $\chi$ can be a simple low-norm distribution~\cite{Micci[13]}. Note that the seed and error vectors in the definitions can be replaced by matrices of appropriate dimensions, that are sampled from the same distributions as the vectors. Such interchange does not affect the hardness of LWE~\cite{Pie[12]}.

\subsection{Learning with Rounding}\label{LWR}
Based on a conjectured hard-to-learn function, Naor and Reingold~\cite{Naor[9]} proposed synthesizers as a foundation to construct PRFs. At first glance, using LWE as the hard learning problem looks a valid option but Naor and Reingold's synthesizer requires a deterministic hard-to-learn function, whereas LWE's hardness depends the random, independent errors that are deliberately added to every output. In fact, without any error, LWE becomes trivially easy to learn. So, the main obstacle in constructing efficient lattice/LWE-based PRFs was finding a way to introduce (sufficiently independent) error terms into each of the exponentially many function outputs, while still keeping the function deterministic. 

Banerjee et al.~\cite{Ban[12]} introduced the LWR problem, in which instead of adding a small random error as done in LWE, a \textit{deterministically} rounded version of the sample is released. In particular, for some $p < q$, the elements of $\mathbb{Z}_q$ are divided into $p$ contiguous intervals of roughly $q/p$ elements each. The rounding function is defined as: $\lfloor \cdot \rceil_p: \mathbb{Z}_q \rightarrow \mathbb{Z}_p$, that maps $x \in \mathbb{Z}_q$ into the index of the interval that $x$ belongs to. Note that the error is introduced only when $q > p$, with ``absolute'' error being roughly equal to $q/p$, resulting in the ``error rate'' (relative to $q)$ to be on the order of $1/p$. For a security parameter $\lambda$, they defined the rounding function $\lfloor \cdot \rceil: \mathbb{Z}_q \rightarrow \mathbb{Z}_p$, where $q \geq p \geq 2$, as:
\[\lfloor x \rceil_p = \left\lfloor \frac{p}{q} \cdot x \right\rceil.\]

That is, if $\lfloor x \rceil_p = i$, then $i \cdot \lfloor q/p \rceil$ is the integer multiple of $\lfloor q/p \rceil$ that is nearest to $x$. So, $x$ is deterministically rounded to the nearest element of a sufficiently ``coarse'' public subset of $p \ll q$, well-separated values in $\mathbb{Z}_q$ (e.g., a subgroup). Thus, the ``error term'' comes solely from deterministically rounding $x$ to a relatively nearby value in $\mathbb{Z}_p$. The problem of distinguishing such rounded products from uniform samples is called the decision-learning with rounding problem, abbreviated as decision-LWR${}_{n,q,p}$. Banerjee et al.\ proved decision-LWR${}_{n,q,p}$ to be as hard as decision-LWE for a setting of parameters where the modulus and modulus-to-error ratio are super-polynomial. 

\begin{definition}
	\emph{Let $n \geq 1$ be the security parameter and moduli $q \geq p \geq 2$ be integers.
		\begin{itemize}
			\item For a vector $\textbf{s} \in \mathbb{Z}^n_q$, define the LWR distribution $L_\textbf{s}$ to be the distribution over $\mathbb{Z}^n_q \times \mathbb{Z}_p$ obtained by choosing a vector $\textbf{a} \leftarrow \mathbb{Z}^n_q$ uniformly at random, and outputting $(\textbf{a},b = \lfloor \langle \textbf{a},\textbf{s} \rangle \rceil_p).$
		\end{itemize}
	}	
\end{definition}

For a given distribution over $\textbf{s} \in \mathbb{Z}^n_q$ (e.g., the uniform distribution), the decision-LWR${}_{n,q,p}$ problem is to distinguish (with advantage non-negligible in $n)$ between any desired number of independent samples $(\textbf{a}_i,b_i) \leftarrow L_\textbf{s}$, and the same number of samples drawn uniformly and independently from $\mathbb{Z}^n_q \times \mathbb{Z}_p$.

\begin{theorem}[\cite{Ban[12]}]
	Let $\chi$ be any efficiently sampleable $B$-bounded distribution over $\mathbb{Z}$, and let $q \geq p \cdot B \cdot \lambda(1)$, with $\lambda$ being the security parameter. Then for any distribution over the secret $\textbf{s} \in \mathbb{Z}^n_q$ , solving the decision-LWR${}_{n,q,p}$ problem is at least as hard as solving decision-LWE${}_{n, q, \chi}$ for the same distribution over $\textbf{s}$.
\end{theorem}

Alwen et al.~\cite{Alwen[13]} gave a new reduction that works for a larger range of parameters, allowing for a polynomial modulus and modulus-to-error ratio. Bogdanov et al.~\cite{Andrej[16]} gave a more relaxed and general version of the theorem proved in~\cite{Alwen[13]}. LWR has been used to construct efficient pseudorandom generators and functions~\cite{Ban[12],Ban[14]}. Prior to the introduction of LWR, all constructions of weaker primitives such as symmetric authentication protocols~\cite{Hopper[01],Juels[05],Katz[10]}, randomized weak PRFs~\cite{Benny[09]}, and message-authentication codes~\cite{Kiltz[11],Pie[12]} from noisy-learning problems were inherently randomized functions, where security relies on introducing fresh noise at every invocation. Due to its ease of use and efficiency, several schemes, such as Saber~\cite{Jan[18]} and Round5~\cite{Hayo[19]}, along with some homomorphic encryption solutions~\cite{Ana[17]}, have based their hardness on LWR.

\subsection{LWE-Based KH-PRFs}\label{foll}
Due to small error being involved, LWE based KH-PRF constructions~\cite{Boneh[13],Ban[14]} only achieve `almost homomorphism', which is defined as:
\begin{definition}[\cite{Boneh[13]}]
	\emph{Let $F: \mathcal{K} \times \mathcal{X} \rightarrow \mathbb{Z}^m_p$ be an efficiently computable function such that $(\mathcal{K}, \oplus)$ is a group. We say that the tuple $(F, \oplus)$ is a $\gamma$-almost key homomorphic PRF if the following two properties hold:}
	\begin{enumerate}
		\item \emph{$F$ is a secure pseudorandom function.}
		\item \emph{For every $k_1, k_2 \in \mathcal{K}$ and every $x \in \mathcal{X}$, there exists a vector $\textbf{e} \in [0, \gamma]^m$ such that: $F_{k_1}(x) + F_{k_2}(x) = F_{k_1 \oplus k_2}(x) + \textbf{e} \bmod p$.}
	\end{enumerate}
\end{definition}
Boneh et al.~\cite{Boneh[13]} gave the first standard-model constructions of KH-PRFs using lattices/LWE. They proved their PRF to be secure under the $m$-dimensional (over an $m$-dimensional lattice) LWE assumption, for error rates $\alpha = m^{-\mathrm{\Omega}(l)}$. Later, Banerjee and Peikert~\cite{Ban[14]} gave KH PRFs from substantially weaker LWE assumptions, e.g., error rates $\alpha = m^{-\mathrm{\Omega}(\log l)}$, yielding improved performance. 

\section{Novel PRF Classes: Definitions}\label{sec4}
In this section, we formally define KIH-PRF, HVL-KIH-PRF and HVL-KH-CPRF. For convenience, our definitions assume seed and error matrices instead of vectors. Note that such interchange does not affect the hardness of LWE~\cite{Pie[12]}.\\[1.5mm]
\textbf{Notations.} We begin by defining the important notations.
\begin{enumerate}[topsep=1mm]
	\item $x = x_{\ell} || x_{r}$, where $1 \leq |x_{\ell}| \leq \lfloor |x|/2 \rfloor$ and $|x_{r}| = |x| - |x_{\ell}|$.
	\item $x_a = x_{a.\ell} || x_{a.r}$, where $1 \leq |x_{a.\ell}| \leq \lfloor |x_a|/2 \rfloor$ and $|x_{a.r}| = |x_a| - |x_{a.\ell}|$.
\end{enumerate}
\textbf{Assumption.} Since the new PRF classes exhibit partial input homomorphism, without loss the generality, we assume $x_r$ to be the homomorphic portion of the input $x$, with $x_\ell$ being the static/fixed portion. Obviously, the definitions remain valid if these are swapped, i.e., if $x_\ell$ is taken to be the homomorphic portion of the input with fixed/static $x_r$. 

\subsection{KIH-PRF}
\begin{definition}\label{Prem} 
	\emph{Let $F: \mathcal{K} \times \mathcal{X'} \times \mathcal{X} \rightarrow \mathbb{Z}^{m \times m}_p$ be a PRF family, such that $(\mathcal{K}, \oplus)$ and $(\mathcal{X}, \ominus)$ are groups. We say that the tuple $(F, \ominus, \oplus)$ is a $\gamma$-almost fully key and partially input homomorphic PRF if the following condition holds:}
	\begin{itemize}
		\item \emph{For every $k_1, k_2 \in \mathcal{K}$, with $x_{1.r}, x_{2.r} \in \mathcal{X}$ and $x_{1.\ell}, x_{2.\ell} \in \mathcal{X'}$, such that $x_{1.\ell} = x_{2.\ell} = x_{\ell}$, there exists a vector $\textbf{E} \in [0,\gamma]^{m \times m}$ such that: 
		\[F_{k_1}(x_{1.\ell} || x_{1.r}) + F_{k_2}(x_{2.\ell} || x_{2.r}) + \textbf{E} = F_{k_1 \oplus k_2}(x_\ell || x_r) \bmod p,\] 
		where $x_r = x_{1.r} \ominus x_{2.r}$.}  
	\end{itemize}
\end{definition}

\subsection{HVL-KIH-PRF}

\begin{definition}\label{PremDef} 
	\emph{Let $\mathcal{X} \subset \mathcal{Y}$, with $\ominus$ defining the surjective mapping: $\mathcal{X} \ominus \mathcal{X} \rightarrow \mathcal{Y}$. Let $F: \mathcal{K} \times \mathcal{X'} \times \mathcal{X} \rightarrow \mathbb{Z}^{m \times m}_p$ and $F': \mathcal{K} \times \mathcal{X'} \times \mathcal{Y} \rightarrow \mathbb{Z}^{m \times m}_p$ be two PRF families, where $(\mathcal{K}, \oplus)$ is a group. We say that the tuple $(F, \ominus, \oplus)$ is a $\gamma$-almost fully key and partially input homomorphic PRF with homomorphically induced variable input length, if the following condition holds:}
\begin{itemize}[topsep=1mm]
\item \emph{For every $k_1, k_2 \in \mathcal{K}$, and with $x_{1.r}, x_{2.r} \in \mathcal{X}$ and $x_{1.\ell}, x_{2.\ell} \in \mathcal{X'}$, such that $x_{1.\ell} = x_{2.\ell} = x_{\ell}$, there exists a vector $\textbf{E} \in [0,\gamma]^{m \times m}$ such that: 
\[F_{k_1}(x_{1.\ell} || x_{1.r}) + F_{k_2}(x_{2.\ell} || x_{2.r}) + \textbf{E} = F'_{k_1 \oplus k_2}(x_\ell, y) \bmod p,\]
where $y = x_{1.r} \ominus x_{2.r}$.}
\end{itemize}
\end{definition}

\subsection{Left/Right KH-CPRF with HVL}

\begin{definition}\label{KH-CPRF-Def} 
	\emph{(Left KH-CPRF with HVL:) Let $\mathcal{X} \subset \mathcal{Y}$, with $\ominus$ defining the surjective mapping: $\mathcal{X} \ominus \mathcal{X} \rightarrow \mathcal{Y}$. Let $F: \mathcal{K} \times \mathcal{W} \times \mathcal{X} \rightarrow \mathbb{Z}^{m \times m}_p$ and $F': \mathcal{K} \times \mathcal{W} \times \mathcal{Y} \rightarrow \mathbb{Z}^{m \times m}_p$ be two PRF families, where $(\mathcal{K}, \oplus)$ is a group. We say that the tuple $(F, \ominus, \oplus)$ is a left key homomorphic constrained-PRF with homomorphically induced variable input length, if for any $k_0 \in \mathcal{K}$ and a fixed $w \in \mathcal{W}$, given $F_{k_0}(w || x) \in \mathcal{F}$, where $x \in \mathcal{X}$, there exists an efficient algorithm to compute $F'_{k_0 \oplus k_1}(w, y) \in \mathcal{F'}$, for all $k_1 \in \mathcal{K}$ and $y \in \mathcal{Y}$.\\[1.5mm]
	(Right KH-CPRF with HVL:) Let $\mathcal{X} \subset \mathcal{Y}$, with $\ominus$ defining the surjective mapping: $\mathcal{X} \ominus \mathcal{X} \rightarrow \mathcal{Y}$. Let $F: \mathcal{K} \times \mathcal{X} \times \mathcal{W} \rightarrow \mathbb{Z}^{m \times m}_p$ and $F': \mathcal{K} \times \mathcal{Y} \times \mathcal{W} \rightarrow \mathbb{Z}^{m \times m}_p$ be two PRF families, where $(\mathcal{K}, \oplus)$ is a group. We say that the tuple $(F, \ominus, \oplus)$ is a right key homomorphic constrained-PRF with homomorphically induced variable input length, if for any $k_0 \in \mathcal{K}$ and a fixed $w \in \mathcal{W}$, given $F_{k_0}(x || w) \in \mathcal{F}$, where $x \in \mathcal{X}$, there exists an efficient algorithm to compute $F'_{k_0 \oplus k_1}(y, w) \in \mathcal{F'}$, for all $k_1 \in \mathcal{K}$ and $y \in \mathcal{Y}$.}
\end{definition}

\section{LWE-Based HVL-KIH-PRF Construction}\label{sec5}
In this section, we present the first construction for a HVL-KIH-PRF family. Our construction is based on the LWE problem, and is inspired from the KH-PRF construction by Banerjee and Peikert~\cite{Ban[14]}. 

\subsection{Rounding Function} Let $\lambda$ be the security parameter. Define a rounding function, $\lfloor \cdot \rceil: \mathbb{Z}_q \rightarrow \mathbb{Z}_p$, where $q \geq p \geq 2$, as:
	\[\lfloor x \rceil_p = \left\lfloor \frac{p}{q} \cdot x \right\rceil.\]
That is, if $\lfloor x \rceil_p = i$, then $i \cdot \lfloor q/p \rceil$ is the integer multiple of $\lfloor q/p \rceil$ that is nearest to $x$. So, $x$ is deterministically rounded to the nearest element of a sufficiently ``coarse'' public subset of $p \ll q$, well-separated values in $\mathbb{Z}_q$ (e.g., a subgroup). Thus, the ``error term'' comes solely from deterministically rounding $x$ to a relatively nearby value in $\mathbb{Z}_p$. As described in Section~\ref{LWR}, the problem of distinguishing such rounded products from uniform samples is called the decision-learning with rounding (LWR${}_{n,q,p})$ problem. The rounding function is extended component wise to vectors and matrices over $\mathbb{Z}_q$.
\subsection{Definitions} Let $l = \lceil \log q \rceil$ and $d = l+1$. Define a gadget vector as:
	\[\textbf{g} = (0,1,2,4,\dots, 2^{l-1}) \in \mathbb{Z}^d_q.\]
Define a deterministic decomposition function $\textbf{g}^{-1}: \mathbb{Z}_q \rightarrow \{0,1\}^d$, such that $\textbf{g}^{-1}(a)$ is a ``short'' vector and $ \forall a \in \mathbb{Z}_q$, it holds that: $\langle \textbf{g}, \textbf{g}^{-1}(a) \rangle = a$, where $\langle \cdot \rangle$ denotes the inner product. The function $\textbf{g}^{-1}$ is defined as:
\[\textbf{g}^{-1}(a) = (x', x_0, x_1, \dots, x_{l-1}) \in \{0,1\}^d,\]
where $x' = 0$, and $a = \sum\limits^{l-1}_{i=0} x_i 2^i$ is the binary representation of $a$. The gadget vector is used to define the gadget matrix $\textbf{G}$ as:
\[\textbf{G} = \textbf{I}_n \otimes \textbf{g} = diag(\textbf{g}, \dots, \textbf{g}) \in \mathbb{Z}^{n \times nd}_q,\]  
where $\textbf{I}_n$ is the $n \times n$ identity matrix and $\otimes$ denotes the Kronecker product. The binary decomposition function, $\textbf{g}^{-1}$, is applied entry-wise to vectors and matrices over $\mathbb{Z}_q$. Thus, $\textbf{g}^{-1}$ is extended to get another deterministic decomposition function $\textbf{G}^{-1}: \mathbb{Z}^{n \times m}_q \rightarrow \{0,1\}^{nd \times m}$, such that, $\textbf{G} \cdot \textbf{G}^{-1}(\textbf{A}) = \textbf{A}$. The addition operations \textit{inside} the binary decomposition functions $\textbf{g}^{-1}$ and $\textbf{G}^{-1}$ are performed as simple integer operations (over all integers $\mathbb{Z}$), and not done in $\mathbb{Z}_q$.\label{Prop}

\subsection{Main Construction} 
The following notations are frequently used throughout this section. 
\begin{itemize}
	\item $x_{\ell h}$: left half of $x$, such that $|x_{\ell h}| = \lfloor |x|/2 \rfloor$,
	\item $x_{rh}$: right half of $x$, such that $|x_{rh}| = \lceil |x|/2 \rceil$,
	\item  $x[i]$: $i^{th}$ bit of bit-string $x$.
\end{itemize}
Let $T$ be a full binary tree with at least one node, with $T.r$ and $T.\ell$ denoting its right and left subtree, respectively. For random matrices $\textbf{A}_0, \textbf{A}_1 \in \mathbb{Z}^{n \times nd}_{q}$, define function $\textbf{A}_T: \{0,1\}^{|T|} \rightarrow \mathbb{Z}^{n \times nd}_q$ recursively as:
\begin{align*}
\textbf{A}_T(x) &= 
\begin{cases}
\textbf{A}_x \qquad \qquad \qquad \qquad \qquad \qquad \quad \;  \:\text{if } |T| = 1 \\
\textbf{A}_{T.\ell}(x_{\ell}) + \textbf{A}_{x[0]} \textbf{G}^{-1}(\textbf{A}_{T.r}(x_{r})) \quad \text{otherwise},
\end{cases}
\end{align*}
where $x = x_{\ell} || x_{r}$, $x_{\ell} \in \{0,1\}^{|T.\ell|}, x_{r} \in \{0,1\}^{|T.r|}$, and $|T|$ denotes the number of leaves in $T$. Based on the random seed $\textbf{S} \in \mathbb{Z}^{n \times nd}_{q}$, the KIH-PRF family, $\mathcal{F}_{(\textbf{A}_0, \textbf{A}_1, T, p)}$, is defined as: 
\[\mathcal{F}_{(\textbf{A}_0, \textbf{A}_1, T, p)} = 
\left\lbrace F_\textbf{S}: \{0,1\}^{2|T|} \longrightarrow \mathbb{Z}^{nd \times nd}_p \right\rbrace.\]
Two \textit{seed dependent} matrices, $\textbf{B}_0, \textbf{B}_1 \in \mathbb{Z}^{n \times nd}_q$, are defined as: 
\[\textbf{B}_0 = \textbf{A}_0 + \textbf{S}; \qquad \qquad \textbf{B}_1 = \textbf{A}_{1} + \textbf{S},\]
Using the seed dependent matrices, a function $\textbf{B}^{\textbf{S}}_T(x)$ is defined recursively as: 
\begin{align*}
\textbf{B}^{\textbf{S}}_T(x) &= 
\begin{cases}
\textbf{B}_x \qquad \qquad \qquad \qquad \qquad \qquad \; \, \qquad \text{if } |T| = 1 \\
\textbf{B}^{\textbf{S}}_{T.\ell}(x_{\ell}) + \textbf{A}_{x[0]} \textbf{G}^{-1}(\textbf{B}^{\textbf{S}}_{T.r}(x_{r})) \qquad \text{otherwise},
\end{cases}
\end{align*}
Let $R: \{0,1\}^{|T|} \rightarrow \mathbb{Z}_q^{nd \times n}$ be a pseudorandom generator. Let $y = y_{\ell h} || y_{rh}$, where $y_{\ell h},y_{rh} \in \{0,1\}^{|T|}$. In order to keep the length the equations in check, we represent the product $R(y_{\ell h}) \cdot \textbf{A}_{y[0]}$ by the \textit{notation:} $R_0(y_{\ell h})$. A member of the KIH-PRF family is indexed by the seed $\textbf{S}$ as:
\begin{equation}\label{deffunc1}
F_{\textbf{S}}(y) := \lfloor \textbf{S}^T \cdot \textbf{A}_{T}(y_{\ell h}) + R_0(y_{\ell h}) \cdot \textbf{G}^{-1} (\textbf{B}^{\textbf{S}}_{T}(y_{rh})) \rceil_p. 
\end{equation}
Let $\bar{0} = 00$, i.e., it represents two consecutive $0$ bits. We define the following function family:
\[
\mathcal{F'}_{(\mathbb{A},T,p)} = 
\left\lbrace F'_{\textbf{S}}: \{0,1\}^{|T|} \times \{0,1,\bar{0}\}^{|T|} \longrightarrow \mathbb{Z}^{nd \times nd}_p \right\rbrace,
\]
where $\mathbb{A} = \{\textbf{A}_0, \textbf{A}_1, \textbf{B}_0, \textbf{B}_1, \textbf{C}_0, \textbf{C}_1, \overline{\textbf{C}}_0\}$, and the matrices $\textbf{C}_1, \textbf{C}_0, \overline{\textbf{C}}_0$ are defined by the seed $\textbf{S} \in \mathbb{Z}^{n \times nd}_q$ as: 
\[\textbf{C}_1 =\textbf{A}_0 + \textbf{B}_1; \quad \overline{\textbf{C}}_0 = \textbf{A}_0 + \textbf{B}_0; \quad \textbf{C}_0 = \textbf{A}_1 + \textbf{B}_1.\]
Define a function $\textbf{C}_T: \{0,1\}^{|T|} \times \{0,1,\bar{0}\}^{|T|} \rightarrow \mathbb{Z}^{n \times nd}_q$ recursively as:
\begin{align*}
\textbf{C}^{\textbf{S}}_T(x) &= 
\begin{cases}
\textbf{C}_x \qquad \qquad \qquad \qquad \qquad \quad \; \qquad  \text{if } |T| = 1 \\
\overline{\textbf{C}}_0 \qquad \qquad \qquad \qquad \qquad \qquad \quad \; \text{if } |T| > 1 ~\bigwedge~ x[i] = x[i+1] = 0\\ 
\textbf{C}^{\textbf{S}}_{T.\ell}(x_{\ell}) + \textbf{A}_{x[0]} \textbf{G}^{-1}(\textbf{C}^{\textbf{S}}_{T.r}(x_{r})) \quad \text{otherwise},
\end{cases}
\end{align*}
i.e., $\overline{\textbf{C}}_0$ denotes two bits. Hence, during the evaluation of $\textbf{C}^\textbf{S}_T(x)$, a leaf in $T$ may represent one bit or two bits. Let $z = z_0 || z_1$, where $z_{0} \in \{0,1\}^{|T|}$ and $z_{1} \in \{0,1,\bar{0}\}^{|T|}$. A member of the function family $\mathcal{F'}_{(\mathbb{A},T,p)}$ is defined as:
\begin{equation}\label{deffunc}
F'_{\textbf{S}}(z_0, z_1) := \lfloor \textbf{S}^T \cdot \textbf{A}_{T}(z_0) + R_0(z_{0}) \cdot \textbf{G}^{-1} (\textbf{C}^{\textbf{S}}_{T}(z_1)) \rceil_p,
\end{equation}
where $R_0(z_{0}) = R(z_{0}) \cdot \textbf{A}_{z_0[0]}$. Similar to the KH-PRF construction from~\cite{Ban[14]}, bulk of the computation performed while evaluating the PRFs, $F_{\textbf{S}}(x)$ and $F'_{\textbf{S}}(x_0,x_1)$, is in computing the functions $\textbf{A}_T(x), \textbf{B}^{\textbf{S}}_T(x), \textbf{C}^{\textbf{S}}_T(x)$. While computing these functions on an input $x$, if all the intermediate matrices are saved, then $\textbf{A}_T(x'),$\\ $\textbf{B}^{\textbf{S}}_T(x'), \textbf{C}^{\textbf{S}}_T(x')$ can be incrementally computed for a $x'$ that differs from $x$ in a single bit. Specifically, one only needs to recompute the matrices for those internal nodes of $T$ which appear on the path from the leaf representing the changed bit to the root. Hence, saving the intermediate matrices, that are generated while evaluating the functions on an input $x$ can speed up successive evaluations on the related inputs $x'$.

\subsection{Proof of Correctness}
The homomorphically induced variable length (HVL) for our function family follows from the fact that $\{0,1\}^{|T|} \subset \{0,1,\bar{0}\}^{|T|}$. So, we move on to defining and proving the fully key and partially input homomorphic property for our function family. We begin by introducing a commutative binary operation, called `almost XOR', which is denoted by $\bar{\oplus}$ and defined by the truth table given in Table~\ref{tab1}. 
\begin{table}
	\vspace{-4mm}
	\centering
\begin{tabular}{P{1.5cm}P{1cm}}
	\hline \hline
	$1 ~\bar{\oplus}~ 1$ & $= 0$\\
	$0 ~\bar{\oplus}~ 0$ & $= 00$\\ 
	$0~ \bar{\oplus}~ 1$ & $= 1$\\ \hline \hline
\end{tabular}
\caption{\label{tab1}Truth Table for `almost XOR' operation, $\bar{\oplus}$}
\vspace{-6mm}
\end{table}
	
\begin{theorem}\label{thm1}
For any inputs $x, y \in \{0,1\}^{|T|}$ and a full binary tree $|T|$ such that: $x_{\ell h} = y_{\ell h} = z_0$ and $x_{rh} \bar{\oplus} y_{rh} = z_1$, where $z_0 \in \{0,1\}^{|T|}$, and $z_1 \in \{0,1,\bar{0}\}^{|T|}$, the following holds: 
\begin{equation}\label{homoEq}
F'_{(\textbf{S}_1 + \textbf{S}_2)}(z_0,z_1) = F_{\textbf{S}_1}(x) + F_{\textbf{S}_2}(y) + \textbf{E},  
\end{equation}
where $||\textbf{E}||_{\infty} \leq 1.$
\end{theorem}

\begin{proof}
We begin by making an important observation, and arguing about its correctness.
\begin{observation}\label{observ}
	\emph{The HVL-KIH-PRF family, defined in Equation.~\ref{deffunc1}, requires both addition and multiplication operations for each function evaluation. Hence, adding the outputs of two functions, $F_{\textbf{S}_1}$ and $F_{\textbf{S}_2}$, from the function family $\mathcal{F}$ translates into adding the outputs per node of the tree $T$. As a result, the decomposition function $\textbf{G}^{-1}$ for each node in $T$ takes one of the following three forms:} 
	\begin{enumerate}
		\item \emph{$\textbf{G}^{-1}(\textbf{A}_b + \textbf{A}_{x[0]} \cdot \textbf{G}^{-1}(\cdot))$,}
		\item \emph{$\textbf{G}^{-1}(\textbf{A}_b + \textbf{A}_{x[0]} \cdot \textbf{G}^{-1}(\cdot) + \dots + \textbf{A}_{x_0} \cdot \textbf{G}^{-1}(\cdot))$,} 
		\item \emph{$\textbf{G}^{-1}(\textbf{A}_b),~b \in \{0,1\}$,}   
	\end{enumerate}
	\emph{where $\textbf{G}^{-1}(\cdot)$ denotes possibly nested $\textbf{G}^{-1}$. Note that each term $\textbf{A}_{x[0]} \cdot \textbf{G}^{-1}(\cdot)$ or a summation of such terms, i.e., $\textbf{A}_{x[0]} \cdot \textbf{G}^{-1}(\cdot) + \dots + \textbf{A}_{x[0]} \cdot \textbf{G}^{-1}(\cdot)$, yields some matrix $\breve{\textbf{A}} \in \mathbb{Z}^{n \times nd}_q$. Hence, each decomposition function $\textbf{G}^{-1}$ in the recursively unwound function $F'_{(\textbf{s}_1 + \textbf{s}_2)}(z_0,z_1)$ has at most two ``direct'' inputs/arguments (i.e., $\textbf{A}_b$ and $\breve{\textbf{A}}$).} 
\end{observation}	

Recall from Section~\ref{Prop} that for the ``direct'' arguments of $\textbf{G}^{-1}$, addition operations are performed as simple integer operations (over all integers $\mathbb{Z}$) instead of being done in $\mathbb{Z}_q$. We know from Observation~\ref{observ}, that unwinding of the recursive function $F'_{\textbf{S}}(z_0,z_1)$ yields each binary decomposition function $\textbf{G}^{-1}$ with at most two ``direct'' inputs. We also know that binary decomposition functions $(\textbf{g}^{-1} \text{ and } \textbf{G}^{-1})$ are linear, provided there is no carry bit or there is an additional bit to accommodate the possible carry. Hence, by virtue of the extra bit, $d-l$ (where $l = \lceil \log q \rceil$, and $d=l+1)$, each $\textbf{G}^{-1}$ behaves as a linear function during the evaluation of our function families, i.e., the following holds:
\begin{equation}\label{linearEq}
\textbf{G}^{-1}(\textbf{A}_i + \textbf{A}_j) = \textbf{G}^{-1}(\textbf{A}_i) + \textbf{G}^{-1}(\textbf{A}_j),	
\end{equation}
where $\textbf{G}^{-1}(\textbf{A}_i) + \textbf{G}^{-1}(\textbf{A}_j)$ is component-wise vector addition of the $n$, $d$ bits long bit vectors of the columns, $[\textbf{v}_1, \dots, \textbf{v}_{nd}] \in \mathbb{Z}^{1 \times nd}$, of $\textbf{G}^{-1}(\textbf{A}_i)$ with the $n$, $d$ bits long bit vectors of the columns, $[\textbf{w}_1, \dots, \textbf{w}_{nd}] \in \mathbb{Z}^{1 \times nd}_q$, of $\textbf{G}^{-1}(\textbf{A}_j)$. We are now ready to prove Equation~\ref{homoEq}. Since $y_{\ell h} = x_{\ell h}$, we use $x_{\ell h}$ to represent both, as that helps clarity. Let $\textbf{S} = \textbf{S}_1 + \textbf{S}_2$, then by using Equation~\ref{deffunc}, we can write the LHS of Equation~\ref{homoEq} as: $\lfloor \textbf{S}^T \cdot \textbf{A}_{T}(z_0) + R_0(z_{0}) \cdot \textbf{G}^{-1} (\textbf{C}^{\textbf{S}}_{T}(z_1)) \rceil_p.$\\[1.5mm] 
Similarly, from Equation~\ref{deffunc1}, we get RHS of Equation~\ref{homoEq} equal to: \\[1mm] 
$\lfloor \textbf{S}_1^T \cdot \textbf{A}_{T}(x_{\ell h}) + R_0(x_{\ell h}) \cdot \textbf{G}^{-1}(\textbf{B}^{\textbf{S}_1}_{T}(x_{rh}))\rceil_p + \lfloor \textbf{S}_2^T \cdot \textbf{A}_{T}(x_{\ell h}) + R_0(x_{\ell h}) \cdot \textbf{G}^{-1}(\textbf{B}^{\textbf{S}_2}_{T}(y_{rh})) \rceil_p + \textbf{E}.$ \\[1mm]
We know that $\lfloor a + b \rceil_p = \lfloor a \rceil_p + \lfloor b \rceil_p + e$. We further know that $x_{\ell h} = z_0$, and $R_0(x_{\ell h}) = R_0(z_{0})$. Thus, from Equation~\ref{linearEq}, the RHS can be written as:\\[1mm]
$\lfloor (\textbf{S}_1 + \textbf{S}_2)^T \cdot \textbf{A}_{T}(z_0) + R_0(z_{0}) \cdot \textbf{G}^{-1} (\textbf{B}^{\textbf{S}_1}_{T}(x_{rh}) + \textbf{B}^{\textbf{S}_2}_{T}(y_{rh})) \rceil_p$\\[1mm]
$ = \lfloor \textbf{S}^T \cdot \textbf{A}_{T}(z_0) + R_0(z_{0}) \cdot \textbf{G}^{-1} (\textbf{C}^\textbf{S}_{T}(x_{rh} \bar{\oplus} y_{rh})) \rceil_p$\\[1mm]
$= \lfloor \textbf{S}^T \cdot \textbf{A}_{T}(z_0) + R_0(z_{0}) \cdot \textbf{G}^{-1} (\textbf{C}^\textbf{S}_{T}(z_1) \rceil_p = LHS.$ \qed
\end{proof}

\section{Security Proof}\label{sec6}
The security proofs as well as the time complexity analysis of our construction depend on the tree $T$. Left and right depth of $T$ are respectively defined as the maximum left and right depths over all leaves in $T$. The modulus $q$, the underlying LWE error rate, and the dimension $n$ needed to obtain a desired level of provable security, are largely determined by two parameters of $T$. The first one, called \textit{expansion e(T)}~\cite{Ban[14]}, is defined as: 
\begin{align*}
e(T) &= 
\begin{cases}
0 \qquad \qquad \qquad \qquad \qquad \qquad \qquad \quad \text{if} \: |T| = 1 \\
max\{e(T.\ell)+1, e(T.r)\} \qquad \qquad \quad \text{otherwise}.
\end{cases}
\end{align*}
For our construction, $e(T)$ is the maximum number of terms of the form $\textbf{G}^{-1}(\cdot)$ that get consecutively added together when we unwind the recursive definition of the function $\textbf{A}_T$. The second parameter is called \textit{sequentiality}~\cite{Ban[14]}, which gives the maximum number of nested $\textbf{G}^{-1}$ functions, and is defined as:
\begin{align*}
s(T) &= 
\begin{cases}
0 \qquad \qquad \qquad \qquad \qquad \qquad \qquad \quad \text{if } |T| = 1 \\
max\{s(T.\ell), s(T.r)+1\} \qquad \qquad \quad \text{otherwise}.
\end{cases}
\end{align*}
For our function families, over the uniformly random and independent choice of $\textbf{A}_0, \textbf{A}_1, \textbf{S} \in \mathbb{Z}^{n \times nd}_q$, and with the secret key chosen uniformly from $\mathbb{Z}^n_q$, the modulus-to-noise ratio for the underlying LWE problem is: $q/r \approx (n \log q)^{e(T)}$. Known reductions~\cite{Reg[05],Pei[09],Bra[13]} (for $r \geq 3 \sqrt{n}$) guarantee that such a LWE instantiation is at least as hard as approximating hard lattice problems like GapSVP and SIVP, in the worst case to within $\approx q/r$ factors on $n$-dimensional lattices. Known algorithms for achieving such factors take time exponential in $n/\log(q/r) = \tilde{\Omega}(n/e(T))$. Hence, in order to obtain provable $2^{\lambda}$ (where $\lambda$ is the input length) security against the best known lattice algorithms, the best parameter values are the same as defined for the KH-PRF construction from~\cite{Ban[14]}, which are:
\begin{equation}\label{Eqn5.3}
n = e(T) \cdot \tilde{\Theta}(\lambda) \;  \; \; \text{and} \; \; \; \log q = e(T) \cdot \tilde{\Theta} (1). 
\end{equation}

\subsection{Overview of KH-PRF from~\cite{Ban[14]}}\label{OverBan}
As mentioned earlier, our construction is inspired by the KH-PRF construction from~\cite{Ban[14]}. Our security proofs rely on the security of that construction. Therefore, before moving to the security proofs, it is necessary that we briefly recall the KH-PRF construction from~\cite{Ban[14]}. Although that scheme differs from our KIH PRF construction, certain parameters and their properties are identical. The rounding function $\lfloor \cdot \rceil_p$, binary tree $T$, gadget vector/matrix $\textbf{g}/\textbf{G}$, the binary decomposition functions $\textbf{g}^{-1}/\textbf{G}^{-1}$ and the base matrices $\textbf{D}_0, \textbf{D}_1$ in that scheme are defined similarly to our construction. There is a difference in the definitions of the decomposition functions, which for our construction are defined as (see Section~\ref{Prop}): $\textbf{g}^{-1}: \mathbb{Z}_q \rightarrow \{0,1\}^d$ and $\textbf{G}^{-1}: \mathbb{Z}^{n \times m}_q \rightarrow \{0,1\}^{nd \times m},$ i.e., the dimensions of the output space for our decomposition functions has $d(=l+1)$ instead of $l = \lceil \log q \rceil$ as in~\cite{Ban[14]}. Recall that the extra (carry) bit ensures that Equation~\ref{linearEq} holds. 
\subsubsection{KH-PRF Construction from~\cite{Ban[14]}}\label{BanSec}
Given two uniformly selected matrices, $\textbf{D}_0, \textbf{D}_1 \in \mathbb{Z}^{n \times nl}_q$, define function $\textbf{D}_T(x): \{0,1\}^{|T|} \rightarrow \mathbb{Z}^{n \times nl}_q$ as:
\begin{equation}\label{BanFunc}
\begin{aligned}
\textbf{D}_T(x) &= 
\begin{cases}
\textbf{D}_x \qquad \qquad \qquad \qquad \qquad \qquad \quad \text{if } |T| = 1 \\
\textbf{D}_{T.\ell}(x_{\ell t}) \cdot \textbf{G}^{-1}(\textbf{D}_{T.r}(x_{rt})) \qquad \; \; \text{otherwise},
\end{cases}
\end{aligned}
\end{equation} 
where $x = x_{\ell t} || x_{rt}$, for $x_{\ell t} \in \{0,1\}^{|T.\ell|}, x_{rt} \in \{0,1\}^{|T.r|}$. The KH-PRF function family is defined as:
\[	\mathcal{H}_{\textbf{D}_0, \textbf{D}_1, T, p} = 
	\left\lbrace H_\textbf{s}: \{0,1\}^{|T|} \rightarrow \mathbb{Z}^{nl}_p \right\rbrace,
\]
where $p \leq q$ is the modulus. A member of the function family $\mathcal{H}$ is indexed by the seed $\textbf{s} \in \mathbb{Z}^n_q$ as: $H_{\textbf{s}}(x) = \lfloor \textbf{s} \cdot \textbf{D}_{T}(x) \rceil_p.$

For the sake of completeness, we recall the main security theorem from~\cite{Ban[14]}.
\begin{theorem}[\cite{Ban[14]}]\label{BanThm}
	Let $T$ be any full binary tree, $\chi$ be some distribution over $\mathbb{Z}$ that is subgaussian with parameter $r > 0$ (e.g., a bounded or discrete Gaussian distribution with expectation zero), and
	\[q \geq p \cdot r \sqrt{|T|} \cdot (nl)^{e(T)} \cdot \lambda^{\omega(1)},\]
	where $\lambda$ is the input size. Then over the uniformly random and independent choice of $\textbf{D}_0, \textbf{D}_1 \in \mathbb{Z}^{n \times nl}_q$, the family $\mathcal{H}_{\textbf{D}_0,\textbf{D}_1,T,p}$ with secret key chosen uniformly from $\mathbb{Z}^n_q$ is a secure PRF family, under the decision-LWE${}_{n,q,\chi}$ assumption.
\end{theorem}

\subsection{Security Proof of Our Construction}
The dimensions and bounds for the parameters $r,q,p,n,m$ and $\chi$ in our construction are the same as in~\cite{Ban[14]}. We begin by defining the necessary terminology.
\begin{enumerate}
	\item Reverse-LWE: is an LWE instance $\textbf{S}^T \textbf{A} + \textbf{E}$ with secret lattice-basis $\textbf{A}$ and public seed matrix $\textbf{S}$. 
	\item Reverse-LWR: is defined similarly, i.e., $\lfloor \textbf{S}^T \textbf{A} \rceil_p$ with secret $\textbf{A}$ and public $\textbf{S}$. 
	\item If $H$ represents the binary entropy function, then we know that for uniformly random $\textbf{A} \in \mathbb{Z}^{n \times nd}_q$ and a random seed $\textbf{S} \in \mathbb{Z}^{n \times nd}_q$, it holds that: $H(\textbf{A}) = H(\textbf{S})$. Hence, it follows from elementary linear algebra that reverse-LWR${}_{n,q,p}$ and reverse-LWE${}_{n,q,\chi}$ are at least as hard as decision-LWR${}_{n,q,p}$ and decision-LWE${}_{n,q,\chi}$, respectively. 
\end{enumerate}

\begin{observation}\label{obs1}
	\emph{Consider the function family $\mathcal{F}_{(\textbf{A}_0, \textbf{A}_1, T,p)}$. We know that a member of the function family is defined by a random seed $\textbf{S} \in \mathbb{Z}^{n \times nd}_q$ as:\\[1.5mm]
	$F_{\textbf{S}}(x) = \lfloor \textbf{S}^T \cdot \textbf{A}_{T}(x_{\ell h}) + R_0(x_{\ell h}) \cdot \textbf{G}^{-1}(\textbf{B}^{\textbf{S}}_{T}(x_{rh})) \rceil_p =  \underbrace{\lfloor \textbf{S}^T \cdot \textbf{A}_{T}(x_{\ell h}) \rceil_p}_{\textbf{L}_T(x_{\ell h})} + \underbrace{\lfloor R_0(x_{\ell h}) \cdot \textbf{G}^{-1}(\textbf{B}^{\textbf{S}}_{T}(x_{rh})) \rceil_p}_{\textbf{R}_T(x_{rh})} + \textbf{E}.$}
\end{observation}

\begin{observation}\label{obs2}
	\emph{For $|T| \geq 1$ and $x \in \{0,1\}^{2|T|}$, each $\textbf{L}_T(x_{\ell h})$ is the sum of the following three types of terms:}
	\begin{enumerate}
		\item \emph{Exactly one term of the form: $\lfloor \textbf{S}^T \cdot \textbf{A}_{x[0]} \rceil_p$, corresponding to the leftmost child of the full binary tree $T$ and the most significant bit, $x[0]$, of $x$.}
		\item \emph{At least one term of the following form: 
			\[\lfloor \textbf{S}^T \cdot \textbf{A}_{x[0]} \cdot \textbf{G}^{-1}(\textbf{A}_{x[i]}) \rceil_p; \qquad (1 \leq i \leq 2|T|),\]
		corresponding to the right child at level $1$ of the full binary tree $T$.}
		\item \emph{Zero or more terms with nested $\textbf{G}^{-1}$ functions of the form: 
			\[\lfloor \textbf{S}^T \cdot \textbf{A}_{x[0]} \cdot \textbf{G}^{-1}(\textbf{A}_{x[i]} + \textbf{A}_{x[0]} \cdot \textbf{G}^{-1}(\cdot) + \dots) \rceil_p; \qquad (1 \leq i \leq 2|T|).\]
		}
	\end{enumerate}
\end{observation}
We prove that for appropriate parameters and input length, each one of the aforementioned terms is a pseudorandom function on its own.
\begin{lemma}\label{lemma1}
	Let $n$ be a positive integer, $q \geq 2$ be the modulus, $\chi$ be a probability distribution over $\mathbb{Z}$, and $m$ be polynomially bounded (i.e. $m = poly(n))$. For a uniformly random and independent choice of $\textbf{A}_0, \textbf{A}_1 \in \mathbb{Z}^{n \times nd}_q$ and a random seed vector $\textbf{S} \in \mathbb{Z}^{n \times nd}_q$, the function family $\lfloor \textbf{S}^T \cdot \textbf{A}_{x[i]} \rceil_p$ for the single bit input $x[i]~(1 \leq i \leq 2|T|)$ is a secure PRF family under the decision-LWE${}_{n,q,\chi}$ assumption. 	
\end{lemma}
\begin{proof}
	We know from~\cite{Ban[12]} that $\lfloor \textbf{S}^T \cdot \textbf{A}_{x[i]} \rceil_p$ is a secure PRF under the decision-LWR${}_{n,q,p}$ assumption, which is at least as hard as solving the decision-LWE${}_{n,q,\chi}$ problem. \qed
\end{proof}
\begin{cor}[To Theorem~\ref{BanThm}]\label{corr1}
	Let $n$ be a positive integer, $q \geq 2$ be the modulus, $\chi$ be a probability distribution over $\mathbb{Z}$, and $m = poly(n)$. For uniformly random and independent matrices $\textbf{A}_0, \textbf{A}_1 \in \mathbb{Z}^{n \times nd}_q$ with a random seed $\textbf{S} \in \mathbb{Z}^{n \times nd}_q$, the function $\lfloor \textbf{S}^T \cdot \textbf{A}_{x[0]} \cdot \textbf{G}^{-1}(\textbf{A}_{x[i]}) \rceil_p$ for the two bit input: $x[0] || x[i]$, is a secure PRF family, under the decision-LWE${}_{n,q,\chi}$ assumption.
\end{cor}
\begin{proof}
	For the two bit input $x[0] || x[i]$, the expression $\lfloor \textbf{S}^T \cdot \textbf{A}_{x[0]} \cdot \textbf{G}^{-1}(\textbf{A}_{x[i]}) \rceil_p$ is an instance of the function $\mathcal{H}_{\textbf{A}_0, \textbf{A}_1, T, p}$ (see Section~\ref{BanSec}). Hence, it follows from Theorem~\ref{BanThm} that $\lfloor \textbf{S}^T \cdot \textbf{A}_{x[0]} \cdot \textbf{G}^{-1}(\textbf{A}_{x[i]}) \rceil_p$ is a secure PRF family under the decision-LWE${}_{n,q,\chi}$ assumption. \qed
\end{proof}
\begin{cor}[To Theorem~\ref{BanThm}]\label{corr2}
	 Let $n$ be a positive integer, $q \geq 2$ be the modulus, $\chi$ be a probability distribution over $\mathbb{Z}$, and $m = poly(n)$. Given uniformly random and independent $\textbf{A}_0, \textbf{A}_1 \in \mathbb{Z}^{n \times nd}_q$, and a random seed $\textbf{S} \in \mathbb{Z}^{n \times nd}_q$, the function: $\lfloor \textbf{S}^T \cdot \textbf{A}_{x[0]} \cdot \textbf{G}^{-1}(\textbf{A}_{x[i]} + \textbf{A}_{x[0]} \cdot \textbf{G}^{-1}(\cdot)+ \dots) \rceil_p$ is a secure PRF family under the decision-LWE${}_{n,q,\chi}$ assumption. 
\end{cor}
\begin{proof}
	Since, $\textbf{A}_0, \textbf{A}_1 \in \mathbb{Z}^{n \times nd}_q$ are random and independent, $\textbf{A}_{x[0]} \cdot \textbf{G}^{-1}(\cdot)$ is statistically indistinguishable from $\textbf{A}_{x[i]} \cdot \textbf{G}^{-1}(\cdot)$, as defined by the function $\textbf{B}_T(x)$ (see Equation~\ref{BanFunc}), where $\textbf{G}^{-1}(\cdot)$ represents possibly nested $\textbf{G}^{-1}$. Hence, it follows from Theorem~\ref{BanThm} that for the ``right spine'' (with leaves for all the left children) full binary tree $T, \lfloor \textbf{S}^T \cdot \textbf{A}_{x[0]} \cdot \textbf{G}^{-1}(\textbf{A}_{x[i]} + \textbf{A}_{x[0]} \cdot \textbf{G}^{-1}(\cdot) + \dots) \rceil_p$ defines a secure PRF family under the decision-LWE${}_{n,q,\chi}$ assumption. \qed
\end{proof}
\begin{cor}[To Theorem~\ref{BanThm}]\label{corr3}
	Let $n$ be a positive integer, $q \geq 2$ be the modulus, $\chi$ be a probability distribution over $\mathbb{Z}$, and $m = poly(n)$. For uniformly random and independent matrices $\textbf{A}_0, \textbf{A}_1 \in \mathbb{Z}^{n \times nd}_q$, and a random seed $\textbf{S} \in \mathbb{Z}^{n \times nd}_q$, the function $\textbf{R}_T(x_{rh}) = \lfloor R_0(x_{\ell h}) \cdot \textbf{G}^{-1}(\textbf{B}^{\textbf{S}}_{T}(x_{rh})) \rceil_p$ is a secure PRF family under the decision-LWE${}_{n,q,\chi}$ assumption.
\end{cor}
\begin{proof}
	We know that $\textbf{A}_0, \textbf{A}_1, \textbf{S} \in \mathbb{Z}^{n \times nd}_q$ are generated uniformly and independently. Therefore, the secret matrices, $\textbf{B}_0, \textbf{B}_1$, defined as: $\textbf{B}_0 = \textbf{A}_0 + \textbf{S}$ and $\textbf{B}_1 = \textbf{A}_{1} + \textbf{S},$ have the same distribution as $\textbf{A}_0, \textbf{A}_1$. As $R: \{0,1\}^{|T|} \rightarrow \mathbb{Z}^{nd \times n}_q$ is a PRG, $R(x_{\ell h})$ is a valid seed matrix for decision-LWE, making $\textbf{R}_T(x_{rh})$ an instance of reverse-LWR${}_{n,q,p}$, which we know is as hard as the decision-LWR${}_{n,q,p}$ problem. Hence, it follows from Theorem~\ref{BanThm} that $\textbf{R}_T(x_{rh})$ defines a secure PRF family for secret $R(x_{\ell h})$. \qed
\end{proof}
\begin{theorem}\label{mainThm}
	Let $T$ be any full binary tree, $\chi$ be some distribution over $\mathbb{Z}$ that is subgaussian with parameter $r > 0$ (e.g., a bounded or discrete Gaussian distribution with expectation zero), $R: \{0,1\}^{|T|} \rightarrow \mathbb{Z}^{nd \times n}_q$ be a PRG, and
	\[q \geq p \cdot r \sqrt{|T|} \cdot (nd)^{e(T)} \cdot \lambda^{\omega(1)}.\]
	Then over the uniformly random and independent choice of $\textbf{A}_0, \textbf{A}_1 \in \mathbb{Z}^{n \times nd}_q$ and a random seed $\textbf{S} \in \mathbb{Z}^{n \times nd}_q$, the family $\mathcal{F}_{(\textbf{A}_0,\textbf{A}_1,T,p)}$ is a secure PRF under the decision-LWE${}_{n,q,\chi}$ assumption.
\end{theorem}
\begin{proof}
	From Observations~\ref{obs1} and~\ref{obs2}, we know that each member $F_\textbf{S}$ of the function family $\mathcal{F}_{(\textbf{A}_0,\textbf{A}_1,T,p)}$ is defined by the random seed $\textbf{S}$, and can be written as:
	\begin{flalign*}
		F_\textbf{S}(x) &= \lfloor \textbf{S}^T \cdot \textbf{A}_{x[0]} \rceil_p + d_1 \cdot \lfloor \textbf{S}^T \cdot \textbf{A}_{x[0]} \cdot \textbf{G}^{-1}(\textbf{A}_{x[i]}) \rceil_p + \\ & d_2 \cdot \lfloor \textbf{S}^T \cdot \textbf{A}_{x[0]} \cdot \textbf{G}^{-1}(\textbf{A}_{x[i]} + \textbf{A}_{x[0]} \cdot \textbf{G}^{-1}(\cdot) + \dots) \rceil_p + \\ & d_3 \cdot \lfloor R(x_{\ell h}) \cdot \textbf{A}_{x[0]} \cdot \textbf{G}^{-1}(\textbf{B}^{\textbf{S}}_{T}(x_{rh})) \rceil_p + \textbf{E},
	\end{flalign*}
	where $d_1, d_2, d_3 \in \mathbb{Z}$, such that, $1 \leq d_1, d_3 \leq |T|$ and $0 \leq d_2 \leq |T|$. From Lemma~\ref{lemma1}, and Corollaries~\ref{corr1},~\ref{corr2} and~\ref{corr3}, we know that the following are secure PRFs under the decision-LWE${}_{n,q,\chi}$ assumption:
	\begin{enumerate}
		\item $\lfloor \textbf{S}^T \cdot \textbf{A}_{x[0]} \rceil_p$, $\lfloor \textbf{S}^T \cdot \textbf{A}_{x[0]} \cdot \textbf{G}^{-1}(\textbf{A}_{x[i]}) \rceil_p$,
		\item $\lfloor \textbf{S}^T \cdot \textbf{A}_{x[0]} \cdot \textbf{G}^{-1}(\textbf{A}_{x[i]} + \textbf{A}_{x[0]} \cdot \textbf{G}^{-1}(\cdot) + \dots) \rceil_p$,
		\item $\lfloor R(x_{\ell h}) \cdot \textbf{A}_{x[0]} \cdot \textbf{G}^{-1}(\textbf{B}_{T}^{\textbf{S}}(x_{rh})) \rceil_p$.
	\end{enumerate} 
	Hence, it follows that the function family $\mathcal{F}_{(\textbf{A}_0,\textbf{A}_1,T,p)}$ is a secure PRF under the decision-LWE${}_{n,q,\chi}$ assumption.    
\end{proof}

\begin{cor}[To Theorem~\ref{mainThm}]\label{PRFcor}
	Let $T$ be any full binary tree, $\chi$ be some distribution over $\mathbb{Z}$ that is subgaussian with parameter $r > 0$ (e.g., a bounded or discrete Gaussian distribution with expectation zero), $R: \{0,1\}^{|T|} \rightarrow \mathbb{Z}^{nd \times n}_q$ be a PRG, and
	\[q \geq p \cdot r \sqrt{|T|} \cdot (nd)^{e(T)} \cdot \lambda^{\omega(1)}.\]
	Then over the uniformly random and independent choice of $\textbf{A}_0, \textbf{A}_1 \in \mathbb{Z}^{n \times nd}_q$ and random seed $\textbf{S} \in \mathbb{Z}^{n \times nd}_q$, the family $\mathcal{F'}_{(\mathbb{A},T,p)}$ is a secure PRF under the decision-LWE${}_{n,q,\chi}$ assumption.
\end{cor}

\section{Time Complexity Analysis}\label{sec7}
In this section, we analyze the asymptotic time complexity of evaluating a function from our HVL-KIH-PRF family. We know that the time complexity of the binary decomposition function $\textbf{g}^{-1}$ is $O(\log q)$, and that $\textbf{G}^{-1}$ is simply $\textbf{g}^{-1}$ applied entry-wise. The size of the public matrices, $\textbf{A}_0, \textbf{A}_1 \in \mathbb{Z}^{n \times nd}_q$, is $\Theta(n^2 \log q)$, which by Equation~\ref{Eqn5.3} is $e(T)^4 \cdot \tilde{\Theta}(\lambda^2)$ bits. The secret matrix, $\textbf{S} \in \mathbb{Z}^{n \times nd}_q$, also has the same size, i.e., $e(T)^4 \cdot \tilde{\Theta}(\lambda^2)$. Computing $\textbf{A}_T(x), \textbf{B}_T(x)$ or $\textbf{C}_T(x)$ requires one decomposition with $\textbf{G}^{-1}$, one $(n \times nd)$-by-$(nd \times nd)$ matrix multiplication and one $(n \times nd)$-by-$(n \times nd)$ matrix addition over $\mathbb{Z}_q$, per internal node of $T$. Hence, the total time complexity comes out to be $\Omega(|T| \cdot n^\omega \log^2 \, q)$ operations in $\mathbb{Z}_q$, where $\omega \geq 2$ is the exponent of matrix multiplication. 

\section{Left/Right HVL-KH-CPRFs}\label{sec8}
In this section, we present the construction of another novel PRF class, namely left/right HVL-KH-CPRFs, as a special case of our HVL-KIH-PRF family. Let $F': \mathcal{K} \times \mathcal{X} \times \mathcal{Y} \rightarrow \mathcal{Z}$ be the PRF defined by Equation~\ref{deffunc}. The goal is to derive a \textit{constrained} key PRF, $k_{x,left}$ or $k_{x,right}$ for every $x \in \mathcal{X}$ and $k_0 \in \mathcal{K}$, such that $k_{x,left} = F_{k_0}(x ||\cdot)$ (where $F: \mathcal{K} \times \mathcal{X} \times \mathcal{X} \rightarrow \mathcal{Z}$ is the PRF family defined by Equation~\ref{deffunc1}) enables the evaluation of the PRF function $F'_{k}(x,y)$ for the key $k = k_0 + k_1$, where $k_1 \in \mathcal{K}$, and the subset of points $\{(x,y): y \in \mathcal{Y}\}$, i.e., all the points where the left portion of the input is $x$. Similarly, the constrained key $k_{x, right} = F_{k_0}(\cdot || x)$ enables the evaluation of the PRF function $F'_{k}(x,y)$ for the key $k = k_0 + k_1$, where $k_1 \in \mathcal{K}$, and the subset of points $\{(y,x): y \in \mathcal{Y}\}$, i.e., all the points where the right side of the input is $x$.\\[2mm]
\textbf{KH-CPRF Construction.} We begin by giving a construction for left KH-CPRF, without HVL, and then turn it into a HVL-KH-CPRF construction. Our HVL-KIH-PRF function, defined in Equation~\ref{deffunc1}, is itself a left KH-CPRF when evaluated as: $F_{k_0}(x_0 || \textbf{1})$, i.e., the key is $k_0 \in \mathcal{K}$, the left side of the input is $x_0 \in \mathcal{X}$, and the right half is an all one vector, $\textbf{1} = \{1\}^{\log |\mathcal{X}|}$. Now, to evaluate $F'_{k}(x_0, x_1)$ at a key $k = k_0 + k_1$, and any right input $x_1 \in \mathcal{Y}$, first evaluate $F_{k_1}(x_0 || x_1')$ and add its output with that of the given constrained function, $F_{k_0}(x_0 || \textbf{1})$, i.e., compute: $F'_{k}(x_0, x_1) = F_{k_1}(x_0 || x_1') + F_{k_0}(x_0 || \textbf{1}),$ where $x'_1 \in \mathcal{X}$, and $x_1 = x_1' \bar{\oplus} \textbf{1}$, with $k = k_0 + k_1$. Recall from Table~\ref{tab1} that `almost XOR', $\bar{\oplus}$, differs from XOR only for the case when both inputs are zero. Hence, having $\textbf{1}$ as the right half effectively turns $\bar{\oplus}$ into $\oplus$, and ensures that all possible right halves $x_1 \in \mathcal{X}$ can be realized via $x'_1 \in \mathcal{X}$. 

Similarly, right KH-CPRF can be realized by provisioning the constrained function $F_{k_0}(\textbf{1} || x_0)$, where $\textbf{1} = \{1\}^{|\log \mathcal{X}|}$ is an all ones vector and $x_0 \in \mathcal{X}$. This interchange allows one to evaluate a different version of our HVL-KIH-PRF function, where the left portion of the input exhibits homomorphism (see Section~\ref{sec4}). Hence, for all $k_1 \in \mathcal{K}$ and $x_1 \in \mathcal{X}$, it supports evaluation of $F'_{k_0+k_1}(x_1 || x_0)$. \\[1.5mm]
\textbf{Achieving HVL.} Wlog, we demonstrate left In order to achieve HVL for the left/right KH-CPRF given above, simply replace the all ones vector, $\textbf{1}$, with an all zeros vector, $\textbf{0}$, of the same dimension. Hence, the new constrained functions become: $k_{left} = F_{k_0}(x_0 || \textbf{0})$ and $k_{right}= F_{k_0}(\textbf{0} || x_0)$. This enables us to evaluate $F'_{k}(x_0, x_1) \in \mathcal{F'}$ at key $k = k_0 + k_1$, for any $k_1 \in \mathcal{K}$ and input $x_1 \in \mathcal{Y}$ as the input homomorphism of the PRF family $F \in \mathcal{F}$ allows us to compute: $F'_{k}(x_0 || x_1) = F_{k_1}(x_0 || x'_1) + F_{k_0}(x_0 || \textbf{0})$, where $x'_1 \in \mathcal{X}$. The pseudorandomness and security follow from that of our HVL-KIH-PRF family.

\section{QPC-UE-UU}\label{sec9}
In this section, we present the first quantum-safe (Q) post-compromise (PC) secure updatable encryption (UE) scheme with unidirectional updates (UU) as an example application of our KIH-HVL-PRF family. An updatable encryption scheme, {\tt UE}, contains algorithms for a data owner and a host. We begin by recalling the definitions of these algorithms. The owner encrypts the data using a secure encryption algorithm, {\tt UE.enc}, and then outsources the ciphertexts to some host. To this end, the data owner initially runs a key generation algorithm, {\tt UE.setup}, to sample an encryption key. This key evolves with epochs and the data is encrypted with respect to a specific epoch $e$, starting with $e = 0$. When moving from epoch $e$ to epoch $e + 1$, the owner invokes the update token algorithm {\tt UE.next} to generate the key material $k_{e+1}$ for the new epoch and calculate the corresponding update token $\Delta_{e+1}$. The owner then sends $\Delta_{e+1}$ to the host, deletes $k_e$ and $\Delta_{e+1}$ immediately, and uses $k_{e+1}$ for encryption from now on. Definition~\ref{defi} gives a formal description of the procedures and algorithms.

\begin{definition}[\cite{Anja[18]}]\label{defi}
	\emph{An updatable encryption scheme {\tt UE} for message space $\mathcal{M}$ consists of a set of polynomial-time algorithms {\tt UE.setup, UE.next, UE.enc, UE.dec}, and {\tt UE.upd} satisfying the following conditions:\\
	{\tt UE.setup}: The algorithm UE:setup is a probabilistic algorithm run by the owner. On input a security parameter $\lambda$, it returns a secret key $k_0 \xleftarrow{\; \$ \;}$  {\tt UE.setup}$(\lambda)$.\\
	{\tt UE.next}: This probabilistic algorithm is also run by the owner. On input a secret key ke for epoch $e$, it outputs the secret key, $k_{e+1}$, and an update token, $\Delta_{e+1}$, for epoch $e + 1$. That is, $(k_{e+1}, \Delta_{e+1}) \xleftarrow{\; \$ \;}$ {\tt UE.next}$(k_e)$.\\
	{\tt UE.enc}: This probabilistic algorithm is run by the owner, on input a message $m \in \mathcal{M}$ and key $k_e$ of some epoch $e$ returns a ciphertext $C_e \leftarrow{\; \$ \;}$ {\tt UE.enc}$(k_e, m)$.\\
	{\tt UE.dec}: This deterministic algorithm is run by the owner, on input a ciphertext $C_e$ and key $k_e$ of some epoch $e$ returns $\{m', \perp\} \leftarrow$ {\tt UE.dec}$(k_e, C_e)$.\\
	{\tt UE.upd}: This probabilistic algorithm is run by the host. Given ciphertext $C_e$ from epoch $e$ and the update token $\Delta_{e+1}$, it returns the updated ciphertext $C_{e+1} \leftarrow$ {\tt UE.upd}$(\Delta_{e+1}, C_e)$. After receiving $\Delta_{e+1}$, the host first deletes $\Delta_{e}$. Hence, during some epoch $e+1$, the update token $\Delta_{e+1}$ is available at the host, but the tokens from earlier epochs have been deleted.}	
\end{definition}

\begin{figure}
	\fbox{\parbox{\linewidth}{
			\vspace{-4mm}
			\begin{itemize}[leftmargin=1em,label={\textbullet},itemsep=0.15cm]
				\item {\tt \textbf{QPC-UE-UU}.setup}$(\lambda)$: Generate a random encryption key $k_0 \xleftarrow{\; \$ \;}$  {\tt F.KeyGen}$(\lambda)$, and sample a random nonce $N_0 \xleftarrow{\; \$ \;} \mathcal{X}$. Set $e \leftarrow 0$, and return the key for epoch $e = 0$ as: $ki_0 = (k_0, N_0)$. 
				\item {\tt \textbf{QPC-UE-UU}.enc}$(ki_e, m)$: Let $i \xleftarrow{\; \$ \;} \mathcal{X}$ be randomly sampled element that is shared by the host and the owner. Parse $ki_e = (k_e, N_e)$, and return the ciphertext for epoch $e$ as: $C_e = (F'_{k_e}(i, N_e) + m)$. Note that the random nonce ensures that the encryption is not deterministic.
				\item {\tt \textbf{QPC-UE-UU}.dec}$(ki_e, C_e)$: Parse $ki_e = (k_e, N_e)$. Return $m \leftarrow C_e - F'_{k_e}(i, N_e)$.
				\item {\tt \textbf{QPC-UE-UU}.next}$(ki_e)$: 
				\begin{enumerate}[leftmargin=1em,itemsep=0mm]
					\item Sample a random nonce $N_{e+1} \xleftarrow{\; \$ \;} \mathcal{X}$. Parse $ki_e$ as $(k_e, N_e)$.
					\item For epoch $e+1$, generate a random encryption key $k_{e+1} \xleftarrow{\; \$ \;}$ {\tt F.KeyGen}$(\lambda)$, and return $\Delta_{e+1} = (\Delta_{e+1}^k, \Delta^N_{e+1})$, where $\Delta^N_{e+1} = N_{e} \bar{\oplus} N_{e+1}$ is the nonce update token, and $\Delta_{e+1}^k = k_{e+1} - k_e$ is the encryption key update token. The key for epoch $e+1$ is $ki_{e+1} = (2 k_e - k_{e+1}, N_{e+1})$, where $2 k_e - k_{e+1}$ is the encryption key.
				\end{enumerate}
				\item {\tt \textbf{QPC-UE-UU}.upd}$(\Delta_{e+1},C_e)$: Update the ciphertext as: $C_{e+1} = C_e - F'_{\Delta_{e+1}^k}(i, \Delta^N_{e+1})$\\ $= F'_{k_e}(i, N_e) + m - F'_{k_{e+1} - k_e}(i, N_{e} \bar{\oplus} N_{e+1})$ \\$= F'_{-k_{e+1} + 2k_e}(i, N_e - (N_e \bar{\oplus} N_{e+1})) + m = F'_{2k_e - k_{e+1}}(i, N_{e+1}) + m$.
			\vspace{-4mm}
	\end{itemize}}}
	\caption{Quantum-safe, post-compromise secure updatable encryption scheme with unidirectional updates, ({\tt \textbf{QPC-UE-UU}}).}\label{UEscheme}
\end{figure}

\subsection{Settings and Notations}\label{settings}
Let $i \xleftarrow{\; \$ \;} \mathcal{X}$ be the identifier for data block $d_i$. Let $F: \mathcal{K} \times \mathcal{X} \times \mathcal{X} \rightarrow \mathcal{Z}$ and $F': \mathcal{K} \times \mathcal{X} \times \mathcal{Y} \rightarrow \mathcal{Z}$ be the functions defined in Equation~\ref{deffunc1} and Equation~\ref{deffunc}, respectively. Let {\tt KeyGen}$(\lambda)$ be the key generation algorithm for $F$, where $\lambda$ is the security parameter. Figure~\ref{UEscheme} gives our {\tt \textbf{QPC-UE-UU}} scheme, wherein a random nonce is generated per key rotation, hence ensuring that the encryption remains probabilistic, despite our PRF family being deterministic. 

\subsection{Proof of Unidirectional Updates}
As explained in Section~\ref{settings}, the random nonce ensures that the encryption in our scheme is probabilistic. Hence, the security of our QPC-UE-UU scheme follows from the pseudorandomness of our HVL-KIH-PRF family. We move on to proving that the ciphertext updates performed by our scheme are indeed unidirectional. Recall that in schemes with unidirectional updates, an update token $\Delta_{e+1}$ can only be used to move ciphertexts from epoch $e$ into epoch $e+1$, but not vice versa. The notations used in the proof are the same as given in Figure~\ref{UEscheme}.

\begin{lemma}
	For the HVL-KIH-PRF family, $\mathcal{F'}: \mathcal{K} \times \mathcal{X} \times \mathcal{Y} \rightarrow \mathcal{Z}$, as defined in Corollary~\ref{PRFcor},z with $\mathcal{X} = \{0,1\}^{|T|}, \mathcal{Y} = \{0,1,\bar{0}\}^{|T|}$ and $\mathcal{Z} = \mathbb{Z}^{nd \times nd}_p$: given ciphertext, $C_{e+1}$ and update token $\Delta_{e+1} = (k_{e+1} - k_e, N_{e} \bar{\oplus} N_{e+1})$ for epoch $e+1$, the following holds for a polynomial adversary $\mathcal{A}$, randomly sampled $Q \xleftarrow{\; \$ \;} \mathcal{Z}$ and security parameter $\lambda$: \vspace{-3mm}
	\[Pr[C_e] = Pr[Q] \pm \epsilon(\lambda),\]
where $\epsilon(\lambda)$ is a negligible function.	
\end{lemma}
\begin{proof}
	The main idea of the proof is that due to the bi-homomorphic property of our HVL-KIH-PRF family, $\mathcal{F'}$, the adversary, $\mathcal{A}$, can only revert back to either the key $k_e$ or the nonce $N_e$, but not both. In other words, $\mathcal{A}$ cannot recover $C_e = F'_{k_e}(i, N_e) + m$. We split the proof into two portions, one concentrating on reverting back to $k_e$ as the function key, and the other one focusing on moving back to $N_e$ as the function input. We prove that these two goals are mutually exclusive for our scheme, i.e., both of them cannot be achieved together.\\[1.5mm]
	\textbf{Case 1: Reverting to $k_e$.} Recall from Table~\ref{tab1} that the only well-defined operation for the operand $\bar{0}$ is $\bar{0} = 0 + 0$. We know that the ciphertext update for epoch $e+1$ is performed as: $C_{e+1} = C_e - F'_{\Delta_{e+1}^k}(i, \Delta^N_{e+1})$. Since $\mathcal{F'}$ is a PRF family, the only way to revert back to $F'_{k_e} \in \mathcal{F'}$ via $C_{e+1}$ and $\Delta_{e+1}$ is by computing: 
	\begin{align*}
		C_{e+1} + F'_{\Delta_{e+1}^k}(i, \Delta^N_{e+1}) &= F'_{2k_e - k_{e+1}}(i, N_{e+1}) + m + F'_{k_{e+1} - k_e}(i, N_e \bar{\oplus} N_{e+1}) \\
														&= F'_{k_e}(i, N_{e+1} \bar{\oplus} (N_e \bar{\oplus} N_{e+1})).
	\end{align*}
	Due to the key homomorphism exhibited by $\mathcal{F'}$, no other computations would lead to the target key $k_e$. We know that $\Delta^N_{e+1} (=N_{e} \bar{\oplus} N_{e+1}) \in \mathcal{Y}$, and that $N_{e+1}, N_e \in \mathcal{X}$. Therefore, the output of the above computation is not well-defined since it leads to $\Delta^N_{e} \bar{\oplus} N_{e+1} \notin \mathcal{Y}$ as being the  input to $F'_{k_e}(i, \cdot) \in \mathcal{F'}$. Hence, when $\mathcal{A}$ successfully reverts back to the target key $k_e$, the nonce deviates from $N_e$ (and the domain $\mathcal{Y}$ itself). \\[1.5mm]
	\textbf{Case 2: Reverting to $N_e$.} Given $C_{e+1}$ and $\Delta_{e+1}$, $\mathcal{A}$ can revert back to $(i, N_e)$ as the function input by computing: $C_{e+1} - F'_{\Delta_{e+1}^k}(i, \Delta^N_{e+1}) = F'_{3k_e - 2k_{e+1}}(i, N_{e}) + m.$ By virtue of the almost XOR operation and the operand $\bar{0}$, the only way to revert back to $N_{e} \in \mathcal{X}$ from $\Delta^N_{e+1} (=N_{e} \bar{\oplus} N_{e+1}) \in \mathcal{Y}$ is via subtraction. But, as shown above, subtraction leads to $F'_{3k_e - 2k_{e+1}}(i, N_e)$ instead of $F'_{k_e}(i, N_e)$. Hence, the computation that allows $\mathcal{A}$ to successfully revert back to $N_e$ as the function input, also leads the function's key to deviate from $k_e$.
	
	Considering the aforementioned arguments, it follows from Corollary~\ref{PRFcor} that $\forall Q \xleftarrow{\; \$ \;} \mathcal{Z}$, it holds that: $Pr[C_e] = Pr[Q] \pm \epsilon(\lambda).$ \qed
\end{proof}

\section{Open Problem: Novel Searchable Encryption Schemes}\label{sec10}
Searchable symmetric encryption (SSE)~\cite{Curt[06]} allows one to store data at an untrusted server, and later search the data for records (or documents) matching a given keyword. Multiple works~\cite{Benny[09],Chang[05],Curt[06],Goh[03],Kamara[12],Kamara[13],Kuro[12],Song[2000],Peter[10]} have studied SSE and provided solutions with varying trade-offs between security, efficiency, and the ability to securely update the data after it has been encrypted and uploaded. A search pattern~\cite{Curt[06]} is defined as any information that can be derived or inferred about the keywords being searched from the issued search queries. In a setting with multiple servers hosting unique shares of the data (generated via threshold secret sharing~\cite{Shamir[79]}), our HVL-KIH-PRF family may be useful in realizing SSE scheme that hides search patterns. For instance, if there are $n$ servers $S_1, S_2, \dots, S_n$, then $n$ random keys $k_1, k_2, \dots, k_n$ can be distributed among them in a manner such that any $t$-out-of-$n$ servers can combine their respective keys to generate $k = \sum_{j=1}^n k_j$. 

If the search index is generated via our HVL-KIH-PRF function $F'_k(i, \cdot) \in \mathcal{F'}$, where $i$ is a fixed database identifier, then to search for a keyword $x$, the data owner can generate a unique, random query $x_j$ for each server $S_j~(1 \leq j \leq n)$. Similar to the key distribution, the data owner sends the queries to the servers such that any $t$ of them can compute $x = \myplus\limits_{j=1}^n x_j$. On receiving search query $x_j$, server $S_j$ uses its key $k_j$ to evaluate $F'_{k_j}(i, x_j)$. If at least $t$ servers reply, the data owner can compute $\sum\limits_{j=1}^t F'_{k_j}(i, x_j) = F'_{k}(i, x)$. Designing a compact search index that does not leak any more information than what is revealed by the PRF evaluation is an interesting open problem, solving which would complete this SSE solution.

\section{Conclusion}\label{sec11}
Key-homomorphic PRFs have found a multitude of interesting applications in cryptography. In this paper, we further expanded the domain of PRFs by introducing bi-homomorphic PRFs, which exhibit full homomorphism over the key and partial homomorphism over the input, which also leads to homomorphically-induced variable input length. We presented a LWE-based construction for such a PRF family, which is inspired by the key-homomorphic PRF construction given by Banerjee and Peikert~\cite{Ban[14]}. 

We use our novel PRF family to develop the first quantum-safe, post-compromise secure updatable encryption scheme with unidirectional updates. The homomorphically-induced variable input length of our PRF family allowed us to address the open problem of achieving unidirectional updates for updatable encryption. As a special case of our PRF family, we introduced key-homomorphic constrained-PRF with homomorphically-induced variable input length. We leave as an open problem the question of using our PRF family to design search pattern hiding searchable symmetric encryption schemes. With appropriately designed search index, the bi-homomorphism property of our PRF family should allow random search queries, hence leading to private search patterns. 

\bibliography{References}{}
\bibliographystyle{plain}

\end{document}